\RequirePackage[orthodox]{nag}

\documentclass[a4paper,USenglish]{lipics-v2021}

\usepackage{xspace}
\usepackage{cite}
\usepackage[ruled,vlined,linesnumbered]{algorithm2e}

\title{An Almost-Optimal Upper Bound on the Push~Number of the Torus Puzzle}

\author{Matteo Caporrella}{Department of Information Engineering, Computer Science, and Mathematics\\University of L'Aquila, Italy}{matteo.caporrella@student.univaq.it}{https://orcid.org/0009-0008-4057-9711}{}
\author{Stefano Leucci}{Department of Information Engineering, Computer Science, and Mathematics\\University of L'Aquila, Italy}{stefano.leucci@univaq.it}{https://orcid.org/0000-0002-8848-7006}{}

\authorrunning{Matteo Caporrella and Stefano Leucci}

\Copyright{Matteo Caporrella and Stefano Leucci}

\relatedversiondetails{Full Version}{https://arxiv.org/abs/2601.08989} 

\supplement{Interactive implementation of the algorithm presented in the paper.}
\supplementdetails{Interactive Resource}{https://www.isnphard.com/g/torus-puzzle/}

\EventEditors{John Iacono}
\EventNoEds{1}
\EventLongTitle{13th International Conference on Fun with Algorithms (FUN 2026)}
\EventShortTitle{FUN 2026}
\EventAcronym{FUN}
\EventYear{2026}
\EventDate{May 18--22, 2026}
\EventLocation{Porquerolles, France}
\EventLogo{}
\SeriesVolume{366}
\ArticleNo{19}

\DeclareMathOperator{\row}{t-row}
\DeclareMathOperator{\col}{t-col}
\DeclareMathOperator{\lsb}{lsb}
\newcommand{\FillColumns}{\textsc{\hbox{FillColumns}}\xspace}
\newcommand{\FloatMinimums}{\textsc{\hbox{FloatMinimums}}\xspace}
\newcommand{\RadixSortBodies}{\textsc{\hbox{RadixSortBodies}}\xspace}
\newcommand{\SwapPairs}{\textsc{\hbox{SwapPairs}}\xspace}

\hideLIPIcs
\nolinenumbers
\begin{document}

\maketitle

\begin{abstract}
    We study the \emph{Torus Puzzle}, a solitaire game in which the elements of an input $m \times n$ matrix need to be rearranged into a target configuration via a sequence of unit rotations (i.e., circular shifts) of rows and/or columns.

    Amano et al.\ proposed a more permissive variant of the above puzzle, where each row and column rotation can shift the involved elements by any amount of positions.
      The number of rotations needed to solve the original and the permissive variants of the puzzle are respectively known as the \emph{push number} and the \emph{drag number}, where the latter is always smaller than or equal to the former and admits an existential lower bound of $\Omega(mn)$.
    While this lower bound is matched by an $O(mn)$ upper bound, 
    the push number is not so well understood. Indeed, to the best of our knowledge, only an $O(mn \cdot \max\{ m, n \})$ upper bound is currently known.

    In this paper, we provide an algorithm that solves the Torus Puzzle using $O(mn \cdot \log \max \{m, n\})$ unit rotations in a model that is more restricted than that of the original puzzle.
    This implies a corresponding upper bound on the push number and reduces the gap between the known upper and lower bounds from $\Theta(\max\{m,n\})$ to $\Theta(\log \max\{m, n\})$.
\end{abstract}

\section{Introduction}
The \emph{Torus Puzzle}, also known as \emph{Sliders}, \emph{TwoBik}, \emph{RotSquare}, \emph{Loopover}, or \emph{Sixteen Puzzle}, is a solitaire game played on an $m \times n$ matrix whose entries consist of (all and only) the integers from $1$ to $mn$. These elements are initially scrambled, and the goal is that of rearranging them via a sequence of \emph{unit rotations} of rows and/or columns that causes the resulting arrangement of the elements to end up in the \emph{sorted configuration}, i.e., in increasing order when read from left to right and from top to bottom.\footnote{More formally, the element on row $i$ and column $j$ of the sorted configuration is the integer $(i-1)n + j$.}
    More precisely, a unit rotation of a row is specified by the row index and by a direction, which can be either leftward or rightward, and consists of a circular shift of all the elements in the chosen row along the chosen direction. 
Similarly, a column can be rotated either upward or downward. Figure~\ref{fig:instance_example} shows an example of such rotations. 

It turns out that not all input matrices are sortable, and it is not hard to derive a characterization of sortable matrices which also serves as a linear-time algorithmic test (see the related works section for a more detailed discussion). 
This motivated previous research to consider only sortable puzzles and to provide general upper and lower bounds on the number of unit rotations needed to solve the puzzle.
This quantity is known as the \emph{push number} and was introduced by Amano et al.\ \cite{AmanoKKKNOTY12} along with the related notion of \emph{drag number}, which counts multiple consecutive rotations of the same row or column only once. Equivalently, the drag number can be interpreted as the number of moves needed to solve the Torus Puzzle when each move allows rotating a row or column of choice by an arbitrary amount.\footnote{To be pedantic, \cite{AmanoKKKNOTY12} only defined the push and drag numbers for a specific sequence of rotations, and then introduced two functions $\text{pn}(A)$ and $\text{dn}(A)$ corresponding to the minimum push and drag numbers among all sequences of rotations that sort the elements of $A$, respectively. We find it natural to refer to $\text{pn}(A)$ and $\text{dn}(A)$ as the push and drag numbers of instance $A$, respectively.}

Amano et al.\ showed a constructive upper bound of $O(mn)$ on the drag number of any sortable $m \times n$ instance of the Torus Puzzle, along with a matching existential lower bound of $\Omega(mn)$,
i.e., that there exists a constant $\varepsilon > 0$ such that, for any choice of $m$ and $n$ with $mn > 1$, some sortable $m \times n$ instance has a drag number of at least $\varepsilon m n$.
Clearly, the push number is no smaller than the drag number, and it can never exceed it by more than a multiplicative factor of $\left\lfloor \frac{\max\{m, n\}}{2} \right\rfloor$ (since a rotation by an arbitrary amount can always be simulated using multiple unit rotations along the ``shortest'' direction).

Unfortunately, to the best of our knowledge, the best asymptotic lower and upper bounds on the push number, in terms of $m$ and $n$, are exactly those obtained by combining the $\Omega(mn)$ and $O(mn)$ bounds on the drag number given in \cite{AmanoKKKNOTY12} with the relations above. 
That is, there exists an algorithm that solves the Torus Puzzle using $O(mn \cdot \max\{m,n\})$ unit rotations, and any such algorithm needs to perform at least $\Omega(mn)$ rotations in the worst case, leaving a multiplicative gap of $\Theta(\max\{m,n\})$ between the upper and the lower bound.

\begin{figure}[t]
    \centering
    \includegraphics{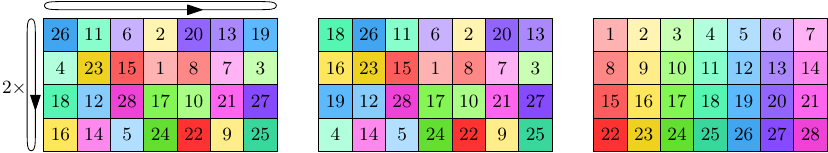}
    \caption{A sortable input instance of the Torus Puzzle (left); the configuration obtained from the input configuration by performing a unit rightward rotation of the first row, followed by two unit downward rotations of the first column (middle); and the target configuration, i.e., the one in which the matrix is sorted (right). Elements belonging to different columns in the target configuration have different hues, while elements belonging to different rows have different lightnesses.}
    \label{fig:instance_example}
\end{figure}

\subparagraph*{Our results and techniques.}
In this paper we make progress towards closing this gap by providing an algorithm solving any sortable $m \times n$ instance of the Torus Puzzle using $O(mn \cdot \log \max \{ m, n \})$ \emph{unit rotations}, which is optimal up to polylogarithmic factors in the worst case. We also show that such a sequence of rotations can be found in the same asymptotic worst-case time.
The above results hold in the even more restricted model in which the rotation directions cannot be freely chosen by the algorithm, but are independently specified, for each row and column, as part of the input instance.

While the solution strategy of \cite{AmanoKKKNOTY12} is based on repeatedly swapping triples of distinct elements, an operation that requires $\Omega(\max\{m, n\})$ unit rotations, our algorithm follows a different approach which relies on three main ingredients.
First, we treat all integers that should end up in the first row of the sorted configuration as some sort of \emph{buffer elements}, whose exact positions can be (temporarily) disregarded;
we move each remaining element to its intended final column, and the buffer elements to the first row.
Next, we sort the non-buffer elements within each column using a strategy akin to that of RadixSort. This requires handling multiple columns in parallel in order to meet the claimed upper bound on the number of moves.
Finally, we deal with the buffer elements: since they are all on the first row, the challenge is that of suitably permuting them without affecting the placement of the other elements. 
As an intermediate step, we show how to simultaneously swap a collection of pairwise-disjoint pairs of buffer elements while causing only controlled side effects to the rest of the matrix. We can then combine multiple parallel swaps in a way that ensures that their side effects cancel out, thus obtaining the sought sorted matrix.  

\subparagraph*{Some historic notes and related works.}

The earliest report we found of the Torus Puzzle dates back to 1981, when Takeshi Ogihara implemented it on a Mitsubishi MELCOM mainframe using a text-based interface. In 1993, following Ogihara’s suggestion while he was involved in upgrading Osaka University’s mainframes to NeXT computers, M. Torigoe, a Faculty of Science student, reimplemented the puzzle as a graphical NeXTSTEP application named RotSquare. Ogihara later ported RotSquare to OpenStep, Mac OS X, and iOS (see Figure~\ref{fig:screenshots_toy}~(left) for a screenshot) \cite{ogihara_personal_comm_2025}. 
Several successive implementations have been independently produced under various names. A commercial physical realization of a puzzle that appears very similar to a $4 \times 8$ Torus puzzle has been reported to exist back in 1987 on the MIT's CubeLovers mailing list (see \cite{
    smith_cubelovers,kornfeld_cubelovers,kinsman_cubelovers} and Figure~\ref{fig:screenshots_toy}~(right) for a 1995 photo). Another puzzle involving a toroidal movement mechanic is implemented by the 1992 video game Yoshi's Cookie for the Nintendo Entertainment System, where players must rotate rows and columns of a grid in order to match similar elements \emph{à la} Bejeweled \cite{juul2007swap,GualaLN14}.

The first mention of the Torus Puzzle in a mathematical context was in 1988, when Diaconis considered the $n \times n$ version as an easier to analyze variant of the Fifteen Puzzle, showing that $O(n^3)$ random rotations are enough to randomize the position of each individual element and conjecturing that $O(n^3 \cdot \log n)$ rotations suffice to shuffle all elements \cite[p.~90]{diaconis1988group}. Diaconis' conjecture is still open to this day, and the best-known upper bound for the \emph{Torus Shuffle} is of $O(n^3 \cdot \log^3 n)$ random rotations \cite{blumberg2024mixing}.

In 1997, mathematician Alexander Bogomolny \cite{Bogomolny1997} described the puzzle on his website under the name \emph{Sliders Puzzle}. He proved that every $4 \times 4$ matrix is sortable by providing an explicit strategy that swaps two elements of choice without affecting the positions of the other elements. In general, each arrangement of the elements of an $m \times n$ matrix can be seen as a permutation of $\{1, 2, \dots, mn\}$, with the sorted configuration corresponding to the identity permutation. Bogomolny observed that, in any $n \times n$ puzzle with odd $n$, all rotations correspond to even permutations, hence all configurations reachable from the initial one must have its same parity. Moreover, all such parity-preserving configurations can actually be reached when $n$ is odd. For even values of $n$, the elements can always be rearranged into any configuration of choice.
Bogomolny's impossibility argument immediately extends to rectangular matrices: an $m \times n$ matrix in which both $m$ and $n$ are odd is sortable only if the permutation induced by the initial configuration of the elements is even.

In 1999, Donald Lee Greenwell gave explicit sequences swapping any two adjacent elements in $6 \times 6$ and $8 \times 8$ matrices
\cite{Greenwell1999}, and argued that similar swaps can be performed in any $n \times m$ matrix, as long as at least one of $n$ and $m$ is even. Greenwell also proposed two additional variants of the puzzle, where rows and columns wrap around as if they were laid out on a Klein bottle or on the projective plane.

Later, Amano et al.\ \cite{AmanoKKKNOTY12} considered general $m \times n$ matrices and described a sequence of rotations performing a three-element swap: given three distinct elements $x_0,x_1,x_2$, it is possible to move each $x_i$ to the position of $x_{(i+1) \bmod 3}$ without affecting the placement of the other elements.
Then, the authors show how a sequence of $O(mn)$ three-element swaps can sort any $m \times n$ matrix where at least one of $m$, $n$, and the initial permutation of the elements is even. 
Since each three-element swap can be implemented by a sequence consisting of a constant number of \emph{compound rotations}\footnote{We use \emph{compound rotation} to refer to multiple consecutive unit rotations of the same row or column.} of rows and columns,
this implies an upper bound of $O(mn)$  on the drag number. However, the overall number of unit rotations required by a three-element swap can be as large as $\Theta(\max\{m, n\})$, yielding the best-known upper bound of $O(mn \cdot \max\{m,n\})$ on the push number.
As observed by \cite{AmanoKKKNOTY12}, a lower bound on the drag number (and on the push number) can be obtained by noticing that there are at least $\frac{(mn)!}{2}$ distinct sortable $m \times n$ matrices 
and at most $(2mn)^{\ell + 1}$ distinct sequences of up to $\ell$ (possibly compound) rotations, implying that some sortable $m \times n$ matrix requires at least $\log_{2mn}\frac{(mn)!}{2} - 1 = \Omega(mn)$ rotations to be sorted.

\noindent To summarize, the results in \cite{Bogomolny1997,Greenwell1999,AmanoKKKNOTY12} characterize the sortable instances of the Torus Puzzle: 
\begin{remark}
    \label{remark:sortable_characterization}
    An $m \times n$ instance $A$ of the Torus Puzzle can be sorted through a sequence of row and column rotations if and only if at least one of $n$, $m$, and the permutation of $\{1, 2, \dots, nm\}$ induced by the arrangement of elements in $A$ is even.
\end{remark}

\begin{figure}
    \centering
    \includegraphics[width=2.90cm]{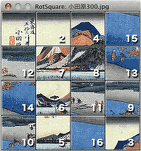}
    \hspace{2cm}
    \includegraphics{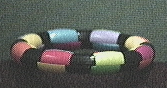}
    \caption{Left: A screenshot of the RotSquare application for Mac OS X by Takeshi Ogihara. Right: a physical puzzle in which eight colored sections of a Torus can rotate in 90-degree increments along the shorter dimension, while the two halves of the torus delimited by a horizontal plane passing through the center can rotate along the long dimension. The puzzle resembles a $4 \times 8$ instance of the Torus Puzzle in which some elements are indistinguishable and two pairs of adjacent rows rotate together. The photo was posted in 1995 to MIT's CubeLovers mailing list \cite{kinsman_cubelovers} and is reproduced here with the author's permission.}
    \label{fig:screenshots_toy}
\end{figure}

Several heuristic strategies for finding solutions to Torus Puzzle instances have been proposed in \cite{HemphillS12}, and a \emph{colored} variant of the Torus Puzzle has been shown to be $\mathsf{NP}$-Hard in \cite{AmanoKKKNOTY12}. In this variant, each entry of the matrix consists of one of two colors, and the goal is that of determining whether some target arrangement of colors can be reached. 

The Torus Puzzle falls into the general category of permutation puzzles \cite{EgnerP98, mulholland2016permutation}: the input configuration is an element of a suitable permutation group, and the goal is that of deciding whether it belongs to the same connected component of the desired target configuration.
If the answer is affirmative, then one can also look for a sequence of moves transforming the initial configuration into the target one, or even for the shortest such sequence. These two problems correspond to the generalized word problem and to the shortest word problem on the associated permutation group.
Some other notable puzzles in this category are the Fifteen Puzzle \cite{archer1999modern,RatnerW86,DemaineR18,Goldreich11}, the Rubik's Cube \cite{DemaineDELW11,DemaineER18} including a variant played on a torus \cite{AlmGR13}, several other \emph{swapping token} problems (see, e.g., \cite{KornhauserMS84,BiloFG024,HikenW25,wilson1974graph} and the references therein), and \emph{rotating puzzles on graphs} \cite{Yang11,MyintUV22}. In particular, the latter appear to be close in spirit to the Torus Puzzle: given a graph $G$ with an element on each vertex, the goal is that of rearranging these elements into a target configuration via a sequence of rotations along the cycles of $G$. An instance of the Torus Puzzle can be seen as a rotating puzzle on an $m \times n$ torus graph, with the further restriction that only some of the cycles can be rotated. 

Finally, we mention that the problem of finding lower and upper bounds to the drag and/or push numbers of Torus Puzzle can be cast as the problem of finding corresponding bounds on the diameter of the Cayley graph of the associated permutation group.

\subparagraph*{Organization of the paper.} 

Section~\ref{sec:preliminaries} introduces the notation used throughout the paper and recalls some standard definitions.
Section~\ref{sec:algorithm} presents our sorting algorithm, which solves any sortable instance of the Torus puzzle using $O(mn \cdot \log \max \{m,n\})$ unit rotations. 
For the sake of simplicity, we assume that the allowed (unit) row rotations are all in the rightward direction, while columns can only be rotated downward.
Section~\ref{sec:conclusions} discusses how the algorithm can be implemented to run in time $O(mn \cdot \log \max\{m, n\})$, how it can be adapted to handle the case in which each row and column can only be rotated in a direction specified as part of the input, and the problems left open.

A playable version of the puzzle, along with an implementation of our algorithm, is available at \url{https://www.isnphard.com/g/torus-puzzle/}, which might help the reader visualize the different subroutines discussed in the technical part of the paper.

\section{Preliminaries}
\label{sec:preliminaries}

\subparagraph*{Matrices and rotations.} We use $A$ to refer to the $m \times n$ matrix to be sorted, and we assume that $m,n \ge 2$.\footnote{Given an $m \times 1$ or a $1 \times n$ matrix $A$, it is easy to design a linear-time algorithm deciding whether $A$ can be sorted and, if that is the case, returning a sequence of $O(nm)$ moves that sorts $A$.} For an integer $x \in \{1, \dots, mn\}$, we denote by $\row(x)$ and $\col(x)$ the \emph{target row} and the \emph{target column} in which $x$ should be placed in the sorted configuration, respectively, i.e., we let
$\row(x) = \lceil x / n \rceil$ and $\col(x) = ((x-1) \bmod {n}) + 1$, where rows and columns are numbered from $1$.

The first row of $A$ will be denoted by $R$, while $R[j]$ will denote either the position on the $j$-th column of $R$ or the element contained therein, depending on context.
Similarly, the $j$-th column of $A$ will be denoted by $C_j$, and the element or position on the $i$-th row of $C_j$ will be denoted by $C_j[i]$.
We refer to $C_j[1]$ and $C_j[m]$ as the \emph{head} and the \emph{tail} of $C_j$, respectively. The \emph{body of $C_j$} refers to all positions/elements that are in column $C_j$ but not in $C_j$'s head (notice that the tail of a column is part of its body). Observe that $R$ contains all and only the columns' heads and, more precisely, $C_j[1]$ coincides with $R[j]$.

Since the model considered in the next section only allows for unit rightward rotations of rows, and unit downward rotations of columns, we refer to a row/column notation by the name of the corresponding row/column, and we may omit the terms ``unit'', ``rightward'', or ``downward''. A non-negative integer superscript $k$ denotes $k$ consecutive rotations of the same row/column, and we may write a sequence of rotations involving multiple rows/columns by juxtaposing row and column names with their respective superscripts. Sequences are evaluated from left to right. For example, $C_1^3 R^2 C_2^0 C_n$ is equivalent to $C_1 C_1 C_1 R R C_n$ and denotes three consecutive downward rotations of $C_1$, followed by two rightward rotations of $R$, followed by a downward rotation of $C_n$.

If $p$ and $q$ are two positions in the matrix, and $S$ is a sequence of rotations, we write $p \xrightarrow{S} q$ to denote that, as a result of applying the sequence $S$, the element initially placed in $p$ ends up in position $q$.
If $p_1, p_2, \dots, p_{k+1}$ are positions and $S_1, S_2, \dots, S_k$ are sequences of rotations such that $p_i \xrightarrow{S_i} p_{i+1}$, we might express this compactly using the shorthand $p_1 \xrightarrow{S_1} p_2 \xrightarrow{S_2} p_3 \xrightarrow{S_3} \dots \xrightarrow{S_k} p_{k+1}$. Observe that the above implies $p_1 \xrightarrow{S_1 S_2 \dots S_k} p_{k+1}$, where $S_1 S_2 \dots S_k$ denotes the concatenation of $S_1, S_2, \dots, S_k$, in this order.

\subparagraph*{Permutations.} Since each rotation can also be interpreted as permutation of the elements in $A$, we employ the \emph{left-to-right} convention for permutation products. 
In this way, the configuration of elements resulting from a sequence of rotations corresponds to the action on $A$ of the permutation resulting from their product.
Finally, we briefly recall some standard concepts used in the paper: a permutation $\pi$ can be decomposed into a product of some number $k$ of disjoint \emph{cycles}\footnote{A cycle $(x_1 \; x_2 \; \dots \; x_\ell )$ of length $\ell \ge 1$ is a permutation $\gamma$ such that: (i) all elements $x_1, \dots, x_\ell$ are distinct, (ii) $\gamma(x_i) = x_{i+1}$ for $i = 1, \dots, \ell-1$, (iii) $\gamma(x_\ell) = x_1$, and (iv) $\gamma(x) = x$ for all $x \not\in \{x_1, \dots, x_\ell\}$.} $C_1, C_2, \dots, C_k$, i.e., $\pi = C_1 C_2 \dots C_k$. A \emph{transposition} is a permutation consisting of a single cycle of length $2$, and an \emph{involution} is a permutation in which each cycle has length at most $2$, i.e., each cycle is either a \emph{fixed-point} (a cycle of length $1$) or a transposition.
The number of cycles of length $2$ of a permutation $\pi$ is denoted by $a_2(\pi)$.
Each permutation $\pi$ can be written as either a product of an even number of transpositions, or as a product of an odd number of transpositions, but not both, and the \emph{parity} of $\pi$ is defined as the parity of the number of such transpositions. If a permutation $\pi$ of $n$ elements has $k$ cycles (including its fixed-points), then the parity of $\pi$ is equal to the parity of $n - k$. Moreover, the parity of a permutation $\pi$ is the same as that of its inverse $\pi^{-1}$.
 The product $\pi \sigma$ of two permutations $\pi$ and $\sigma$ is even if and only if the parities of $\pi$ and $\sigma$ are the same (i.e., $\pi$ and $\sigma$ are both even or both odd).

\section{Solving the Torus Puzzle using \texorpdfstring{$O(mn \cdot \log \max\{m,n\})$}{O(mn log max{m,n})} unit rotations}
\label{sec:algorithm}

In this section we provide an algorithm solving any sortable $m \times n$ instance of the Torus Puzzle using only $O(mn \cdot \log \max\{m, n\})$  rotations. We consider the case in which all rows can only rotate rightward while columns can only rotate downward. In Section~\ref{sec:conclusions} we discuss how arbitrary rotation directions can be handled.

We start with some definitions. We say that a column $C_j$ is \emph{near-full} if all elements $x$ in the body of $C_j$ satisfy $\col(x) = j$, i.e., they are in their target column. If $C_j$ is not near-full, we say that it is \emph{underfull}.
A column $C_j$ is \emph{body-full} if it is near-full and its body contains all (and only) the elements $x$ with $\col(x) = j$ and $\row(x)>1$.
In other words, the body of $C_j$ contains exactly the same set of elements as that in the sorted configuration, but these elements can appear in an arbitrary order.
Finally, a column is \emph{body-sorted} if all elements in its body are in their respective target positions.

The algorithm first makes all columns near-full, then body-full, then sorts their bodies in groups so that all columns become body-sorted, and finally handles sorting the first row. Each of the following subsections describes one of the above steps, in the above order.

\subsection{Making columns near-full}

We start by describing a procedure, which we name \FillColumns, that receives as input an $m \times n$ matrix $A$ with $m \le n$ 
and rearranges the elements in $A$ via a sequence of $O(mn \cdot \log n)$ rotations to ensure that all columns of the resulting configuration are near-full. Intuitively, the algorithm uses the first row as a sort of ``sushi belt'' from which each column can ``pick up'' a desirable element while simultaneously ``discarding'' an undesirable one. 

\begin{algorithm}[t]
    \SetKwInput{Input}{Input}
    \SetKwInput{Output}{Output (in place)}
    \Input{An $m \times n$ matrix $A$ with $m \le n$.}
    \Output{A matrix in which all columns are near-full.}
    \BlankLine
    \caption{Procedure \FillColumns.}
    \label{alg:FillColumns}
    \While{there exists an underfull column}
    {
        \For{$j = 1, 2, \dots, n$}
        {
            \While{$\col(R[j]) = j$ and $C_j$ is underfull}{
                Rotate $C_j$\;
            }
        }
        Rotate $R$\;
    }
\end{algorithm}

In details, $R$ is repeatedly rotated and, whenever (i) some element $x$ in $R$ ends up in its target column $C_j$ (i.e., $\col(x) = j$), and (ii) such a column is underfull, then $C_j$ is also rotated. 
The rotation  of $C_j$ causes $x$ to move from the head to the body of $C_j$, and the tail of $C_j$ to move to the head of $C_j$ (in place of $x$), i.e., $C_j[1] \xrightarrow{C_j} C_j[2]$ and $C_j[m] \xrightarrow{C_j} C_j[1]$.
These column rotations are repeated until either $C_j$ becomes near-full, or the target column of the new head of $C_j$ is no longer $C_j$ (possibly both). 
When no element $x$ in $R$ satisfies conditions (i) and (ii), the procedure resumes rotating the first row. The procedure stops as soon as all columns are near-full. The pseudocode of \FillColumns is shown in Algorithm~\ref{alg:FillColumns}, while Figure~\ref{fig:fillcolumns} shows a possible input, an intermediate configuration, and the resulting output.

\noindent We start our analysis by bounding the number of column rotations.

\begin{lemma}
    \label{lemma:fillcolumns_column_rotations}
    The number of column rotations performed by \FillColumns is $O(mn)$. 
\end{lemma}
\begin{proof}
    We prove the claim by focusing on a generic column $C_j$ and showing that, once $C_j$ is rotated $m-1$ times, \FillColumns performs no additional rotations on $C_j$.
    
   Indeed, \FillColumns ensures that each rotation of $C_j$ moves some element $x$ having target column $C_j$ into the body of $C_j$ and, specifically, into $C_j[2]$. Since $m-1$ additional downward rotations of $C_j$ are needed for $x$ to leave the body of $C_j$, we have that after performing any number $h \in \{0, \dots, m-1\}$ rotations of $C_j$, the body of $C_j$ contains at least $h$ elements whose target column is $C_j$.
   Then, if the $(m-1)$-th rotation of $C_j$ is performed, it must result in a column that is near-full. The claim follows by observing that \FillColumns only rotates columns that are not near-full.
\end{proof}

Let $A_0 = A$ be the input matrix to $\FillColumns$, and let $A_h$ be the matrix as arranged immediately after the completion of the $h$-th iteration of the outer while loop of Algorithm~\ref{alg:FillColumns}, if it exists.
We associate a \emph{potential} $\Psi(A_h)$ to each $A_h$, where $\Psi(A_h)$ is defined as the number of elements $x$ in $A_h$ that are in the body of some column other than their target column, i.e., $\Psi(A_h) = | \{ (i,j) \in \{2, \dots, m\} \times \{1, \dots, n \} : \col(C_j[i]) \neq j \} |$. 
The following lemma shows that each group of $n$ iterations of the outer while loop starting from some configuration $A_h$ decreases the potential by at least half the number of underfull columns in $A_h$.

\begin{lemma}
   \label{lemma:potential_decrease}
    Consider a configuration $A_h$ with $f$ underfull columns. If $A_{h+n}$ exists, then $\Psi(A_{h+n}) \le \Psi(A_h) - \frac{f}{2}$.
\end{lemma}
\begin{proof}
    We assume that $A_{h+n}$ exists since otherwise the claim is trivially true.
    For an underfull column $C_j$ of $A_h$, we let $X_j$ be the set of elements $x$ in $R$ that have $C_j$ as their target column (i.e., such that  $\col(x) = j$), while for each near-full column $C_j$ we let $X_j = \emptyset$.    

    Since row $R$ of configuration $A_h$ contains at most one element with target column $C_j$ for each of the $n-f$ near-full columns $C_j$, there must be at least $f$ distinct elements in $R$ whose target column is underfull. Hence, $\sum_{j=1}^n |X_j| \ge f$.

    In the rest of the proof we argue that at least $\frac{1}{2} \sum_{j=1}^n |X_j|$ elements from $\cup_{j=1}^n X_j$ are placed in the bodies of their respective columns in $A_{n + h}$.
    Since the only way for \FillColumns to move an element $x$ from the head to the body of its target column $C_j$ is that of rotating $C_j$ until some other element $y \neq x$ with $\col(y) \neq j$ ends up in $R$, causing $\Psi(\cdot)$ to decrease by $1$, this will imply that $\Psi( A_{h+n} ) \le \Psi(A_h) - \frac{1}{2} \sum_{j=1}^n |X_j| \le \Psi(A_h) - \frac{f}{2}$, as claimed.

    In particular, we focus on a specific column $C_j$ with $|X_j| \ge 1$, and we show that at least $\max\{ 1, |X_j| - 1 \} \ge \frac{|X_j|}{2}$ of the elements in $X_j$ are in the body of $C_j$ in configuration $A_{h + n}$. 
    Since row $R$ is rotated $n$ times from configuration $A_h$ to configuration $A_{h+n}$, and $C_j$ is only rotated if $\col(C_j[1]) = j$, we have that each element $x \in X_j$ becomes the head of $C_j$ in some (distinct) configuration $A_{\ell(x)}$ with $h \le \ell(x) < h+n$.
    Let $x_1, \dots, x_k$ be the elements in $X_j$, in increasing order of $\ell(x_i)$. 
    Consider then the $\ell(x_i)$-th iteration of the outer while loop of \FillColumns. If $i=1$ or $i < k$, then $C_j$ is underfull in $A_{\ell(x_i)}$, which means that $x$ is moved to the body of $C_j$ in such iteration, and never leaves it in the following iterations. It follows that at least $\max\{ 1, k-1\}$ elements are in the body of $C_j$ in configuration $A_{h+n}$.
\end{proof}

\begin{figure}[t]
    \centering
    \includegraphics{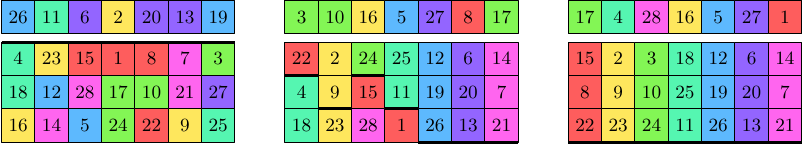}
    \caption{An input instance $A = A_0$ of the Torus Puzzle (left), and the matrices $A_9$ and $A_{17}$ corresponding the arrangement of the elements in $A$ at end of the ninth and of the last iteration of the outer while loop of \FillColumns (middle and right, respectively). Elements are colored according to their target columns. A bold segment on each column $C_j$ indicates the number of rotations of $C_j$ performed by \FillColumns until the shown configuration. A small gap has been added under row $R$ for improved visual clarity. Notice how all columns in $A_{17}$ are near-full.}
    \label{fig:fillcolumns}
\end{figure}

\noindent The above lemma allows us to upper bound the number of rotations of the first row.

\begin{lemma}
    \label{lemma:fillcolumns_row_rotations}
    The number of rotations of $R$ performed by \FillColumns is $O(mn \cdot \log n)$.
\end{lemma}
\begin{proof}
    We partition the rotations of $R$ performed by \FillColumns into $1 + \lfloor \log n \rfloor$ phases.
    We say that a configuration $A_h$ belongs to the $i$-th phase if its number of underfull columns lies in the half-open interval $[\frac{n}{2^{i-1}}, \frac{n}{2^{i}} )$. We also say that a rotation of $R$
belongs to phase $i$ if the configuration $A_h$ from which it is performed belongs to phase $i$. 

    By our definition of $\Psi(\cdot)$, we have that if $A_h$ is in phase $i$, then $\Psi(A_h) \le \frac{n}{2^{i-1}}(m-1)$. Moreover, Lemma~\ref{lemma:potential_decrease} ensures that, if $A_{h+n}$ exists, then $\Psi(A_{h+n}) \le \Psi(A_h) - \frac{n}{2^{i+1}}$.
    This shows that there can be at most $\frac{n}{2^{i-1}} (m-1) \cdot \frac{2^{i+1}}{n} = 4(m-1)$ disjoint groups of $n$ rotations of $R$ in phase $i$. Hence, the number of rotations of $R$ in each phase is at most $(n-1) + 4n(m-1) < 4mn$. The claim follows by multiplying $4mn$ by the number of phases.
\end{proof}

The stopping condition of $\FillColumns$, together with Lemma~\ref{lemma:fillcolumns_column_rotations} and Lemma~$\ref{lemma:fillcolumns_row_rotations}$, imply the following corollary, which summarizes the result of this section.
\begin{corollary}
    \label{cor:near_full}
    Given an $m \times n$ input matrix $A$ with $m \le n$, \FillColumns performs $O(mn\cdot \log n)$ rotations and rearranges the elements in $A$ so that all columns become near-full.
\end{corollary}

\subsection{Transforming near-full columns into body-full columns}

Once all the columns of a configuration $A$ are near-full, it is easy to see that $O(mn)$ rotations suffice to rearrange the elements in $A$ so that all columns become body-full. This is exactly the purpose of our next procedure, which we name \FloatMinimums. The corresponding pseudocode is given in Algorithm~\ref{alg:FloatMinimums}, while Figure~\ref{fig:floatminimums} shows an example input configuration along with the resulting output.

\begin{algorithm}[t]
    \SetKwInput{Input}{Input}
    \SetKwInput{Output}{Output (in place)}
    \Input{An $m \times n$ matrix $A$ in which all columns are near-full.}
    \Output{A matrix in which all columns are body-full.}
    \BlankLine
    \caption{Procedure \FloatMinimums.}
    \label{alg:FloatMinimums}
    \For{$n$ times}{
        \For{$j=1,\dots,n$}{
            \If{$\col(R[j]) = j$}{
                \While{$\row(R[j]) \neq 1$}{
                    Rotate $C_j$;
                }
            }
        }
        Rotate $R$;
    }
\end{algorithm}

We start by noticing that the fact that all columns of $A$ are near-full implies that $R$ contains exactly one element $x_j$ with $\col(x_j) = j$ per column $C_j$.
\FloatMinimums uses this fact by performing $n$ rotations of $R$: before each rotation, it checks whether any element $x_j$ in $R$ is on its target column $C_j$. Whenever this is the case, column $C_j$ contains all (and only) the elements having $j$ as their target column (in some arbitrary order), and \FloatMinimums rotates $C_j$ until its head becomes the unique element $x$ with $\row(x)=1$ and $\col(x) = j$, thus making $C_j$ body-full with at most $m-1$ rotations of $C_j$.
Column $C_j$ will never be rotated again by the procedure. Hence, we have the following:

\begin{lemma}
    \label{lemma:body_full}
    Given an $m \times n$ input matrix $A$ in which all columns are near-full, procedure \FloatMinimums performs $O(mn)$ rotations and rearranges the elements in $A$ so that all columns become body-full.
\end{lemma}

 \subsection{Transforming body-full columns into body-sorted columns}
In this section, we show how a configuration of an $m \times n$ matrix $A$, with $m \le n$, in which all columns are body-full, can be transformed into a configuration in which all columns are body-sorted using $O(mn \cdot \log m)$ rotations.

\subparagraph*{Sorting the body of a single column.}
As a warm-up, we start by sketching how a single body-full column $C_j$ can be made body-sorted using $O(n \cdot \log m)$ rotations. The procedure will only rotate $R$ and $C_j$, hence all elements in the body of columns other than $C_j$ will remain in their initial positions. In particular, any other column that is already body-full (resp.\ body-sorted) will remain body-full (resp.\ body-sorted). This also means that the set of elements in $R$ will remain unchanged, although their order might change.

The idea is that of simulating the RadixSort algorithm on the elements in the body of $C_j$. We achieve this by repeatedly \emph{unloading} the body of $C_j$ into $R$ and then re-loading the same elements back into $C_j$ in a suitable order.
In particular, the procedure consists of $\lceil \log(m-1) \rceil$ phases, where the $\ell$-th phase sorts the elements $x$ of $C_j$ according to the $\ell$-th least significant bit of (a suitably shifted version of) $\row(x)$.
This sorting step is stable; hence, after the last phase is completed, column $C_j$ is body-sorted. 

In details, a phase is subdivided into three sub-phases. In the first sub-phase we unload all elements of $C_j$ into $R$ with the sequence of rotations $(C_j R)^{m-1} R^{n-m+1}$, which also ensures that the overall number of rotations of $R$ is $n$. 
The second sub-phase rotates $R$ $n$ additional times and, before each rotation, we check whether the element $x$ on the head of $C_j$ should be moved to the body of $C_j$, namely if $\col(x) = j$, $\row(x) \ge 2$, and the $\ell$-th least significant bit $\lsb(\row(x)-2, \ell)$ of $\row(x)-2$ is set. 
At this point, $R$ has been rotated $2n$ times, and all the elements that were originally in the body of $C_j$ and have their $\ell$-th least significant bit set are now in the body of $C_j$.
The third and final sub-phase is similar to the second one, with the only difference that the condition  $\lsb(\row(x)-2, \ell) = 1$ is replaced with $\lsb(\row(x)-2, \ell) = 0$. This loads the remaining elements into $C_j$. Observe that in the second and third sub-phases, the elements inserted in $C_j$ are encountered in the same order in which they were unloaded, hence each phase is stable.

Unfortunately, while sequentially applying the above simulation of RadixSort on each column $C_j$ would render all columns body-full, this would require $\Omega(n^2 \cdot \log m)$ rotations, which is asymptotically larger than our $O(mn \cdot \log n)$ goal whenever $n = \omega( m )$. 

\begin{figure}[t]
    \centering
    \includegraphics{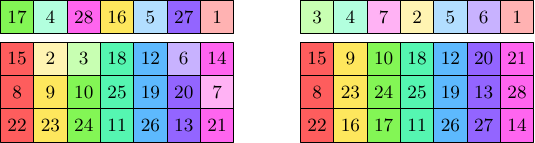}
    \caption{A possible input to \FloatMinimums, which corresponds to the output of \FillColumns in Figure~\ref{fig:fillcolumns} (left), and the resulting output (right). Elements belonging to different columns in the sorted configuration have different hues, and elements $x$ with $\row(x)=1$ are lighter. Notice how all columns of the output are body-full.}
    \label{fig:floatminimums}
\end{figure}

\subparagraph*{Sorting the bodies of up to $\frac{n}{m-1}$ columns.}

To circumvent this issue, we show how to generalize the above RadixSort-inspired scheme to handle a number $k$ of up to $\frac{n}{m-1}$ columns in parallel. The resulting procedure, which we name \RadixSortBodies, performs $O((n+km) \cdot \log m)$ rotations and its pseudocode is given in Algorithm~\ref{alg:RadixSortBodies}.

Let $\mathcal{C} = \{j_1, j_2, \dots, j_k\}$ be a set of $k$ column indices.
\RadixSortBodies starts by unloading, in parallel, the bodies of all columns $C_j$ with $j \in \mathcal{C}$ into $R$.
In order to do that, we choose $k$ pairwise-disjoint subsets $I_1, \dots, I_k$ of $\{ 1, \dots, n \}$, each containing $m-1$ elements, e.g., we can choose $I_h = \{ (h-1)(m-1) + 1, (h-1)(m-1) + 2, \dots, h(m-1) \}$.\footnote{The algorithm described previously for sorting the body of a single column corresponds to the case $k=1$, $\mathcal{C} = \{j_1\}$, and the choice $I_1 = \{1, \dots, j_1\} \cup \{n - m + j_1 + 2, \dots, n\}$ (where the second set of the union might be empty).}
Intuitively, set $I_h$ is associated to column $C_{j_h}$ and contains the $m-1$ indices $i$ corresponding to the positions $R[i]$ into which \RadixSortBodies unloads the elements in the body of $C_{j_h}$. 

As before, the procedure works in $\lceil \log(m-1) \rceil$ phases, where the $\ell$-th phases is concerned with sorting the elements of interest w.r.t.\ their $\ell$-th least significant bit. Similarly, each phase is also split into three sub-phases. 

The first sub-phase unloads all elements in the body of the columns $C_j$ with $j \in \mathcal{C}$ into $R$. In order to do so, we define $\rho(j, r)  = ((j - r - 1) \bmod n) + 1$ so that, as a result of performing $r$ rotations of $R$, the element that ends up in $R[j]$ is the one that was originally placed in $R[\rho(j, r)]$.
In this sub-phase, row $R$ is rotated $n$ times and, before each rotation, we check whether $\rho(j_h, r) \in I_h$ for each $h \in \mathcal{C}$. For all $j_h$ that satisfy the above condition, we rotate $C_{j_h}$, thus moving the element currently on the tail of $C_{j_h}$ into $R$.

\begin{algorithm}[t]
    \caption{Procedure \RadixSortBodies}
    \SetKwInput{Input}{Input}
    \SetKwInput{Output}{Output (in place)}
    \Input{An $m \times n$ matrix $A$ with $m \le n$ such that all columns of $A$ are body-full. A collection $\mathcal{C} = \{j_1, \dots, j_k\}$ of distinct column indices.}
    \Output{A matrix in which all columns $C_j$ with $j \in \mathcal{C}$ are body-sorted, and the body of each column $C_j$ with $j \not\in \mathcal{C}$ is unaffected.} \BlankLine
    \label{alg:RadixSortBodies}
   
    \For{$\ell = 1, 2, \dots, \lceil \log(m-1) \rceil$}
    {
        \tcp{First phase}
        \tcp{Unload the elements in the bodies of the columns in $\mathcal{C}$ into $R$}
        \For{$r = 0,1,\dots, n-1$}{
            \ForEach{$j_h \in \mathcal{C}$}{
                \If{$\rho(j_h, r) \in I_h$}
                {
                    Rotate $C_{j_h}$;
                }
            }
            Rotate $R$;
        }
        \BlankLine

        \tcp{Second and third phases (with $b=1$ and $b=0$, respectively)}
        \tcp{Load the elements back into the bodies of the columns in $\mathcal{C}$}
        \For{$b = 1, 0$}
        {
            \For{$n$ times}{
                \ForEach{$j \in \mathcal{C}$}{
                    \If{$\col(R[j])=j \wedge \row(R[j]) \ge 2 \wedge\lsb(\row(R[j])-2, \ell)=b$}
                    {
                        Rotate $C_j$;
                    }
                }
                Rotate $R$;
            }
        }
    }
\end{algorithm}

The second and third sub-phases behave similarly to one another, and handle the elements having their $\ell$-th least significant bit set to $1$ and $0$, respectively.
The second (resp.\ third) sub-phase, rotates $R$ $n$ times and, before each rotation, checks whether any of the elements currently on the heads the columns whose bodies are to be sorted should be loaded into their column's body. More precisely, if $x$ is the head of some column $C_{j}$ with $j \in \mathcal{C}$, this happens when $\col(x) = j$, $\row(x) \ge 2$, and $\lsb(\row(x)-2, \ell) = 1$ for the second sub-phase (resp.\ $\lsb(\row(x)-2, \ell) = 0$ for the third sub-phase). Notice that, within each of the last two sub-phases, the loaded elements are inserted in the bodies of the respective columns in the same order as they were unloaded in the first phase, hence the overall algorithm is stable. Figure~\ref{fig:radixsortbodies} shows some configurations encountered during an execution of \RadixSortBodies.

\subparagraph*{Sorting all column bodies.}
It is now easy to transform any matrix in which each column is body-full into
a matrix in which each column is body-sorted using $O(mn \cdot \log n)$ rotations.

Indeed, it suffices to arbitrarily partition the $n$ columns into $k = \left\lceil \frac{n}{ \lfloor \frac{n}{m-1} \rfloor } \right\rceil$ groups $G_1, G_2, \dots, G_k$ of $\big\lfloor \frac{n}{m-1} \big\rfloor$ columns each, except possibly for $G_k$ (which has size $n \bmod (m-1)$).
Then, we iteratively make the columns in each group $G_i$ body-full by invoking \RadixSortBodies on the columns in $G_i$, which requires $O\big( (n + \frac{n}{m-1} \cdot m) \cdot \log m  \big) = O(n \cdot \log m)$ rotations.
Since the number of groups is $O(m)$, the overall number of rotations is $O( m n \cdot \log m )$. 

\noindent The following lemma summarizes the discussion in this section.
\begin{lemma}
    \label{lemma:body-sorted}
    Given an $m \times n$ input matrix $A$ with $m \le n$ and in which all columns are body-full, there exists a procedure that performs $O(mn \cdot \log n)$ rotations and rearranges the elements in $A$ so that all columns become body-sorted.
\end{lemma}

\begin{figure}[t]
    \centering
    \includegraphics{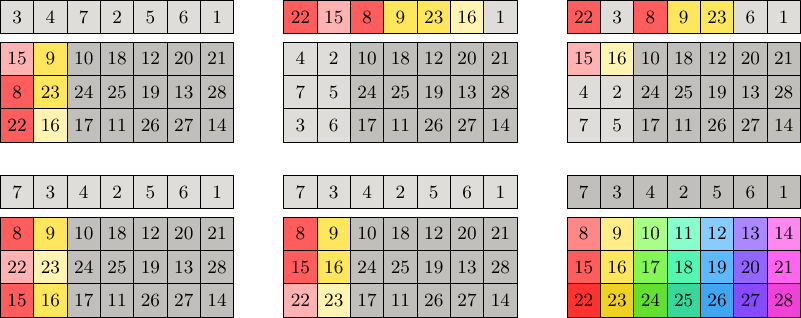}
    \caption{A possible input matrix for an execution \RadixSortBodies sorting $C_1$ and $C_2$ (top left) which coincides to that output by \FloatMinimums in Figure~\ref{fig:floatminimums}, along with the configurations resulting after the first, second, and third sub-phases of the first phase of \RadixSortBodies (top middle, top right, and bottom left, respectively). The remaining two configurations are those at the end of the second phase of the algorithm (bottom middle), and at the end of multiple calls to \RadixSortBodies which sort all column bodies. The elements $x$ that belong to the bodies of $C_1$ and $C_2$ in the input configuration are shown in red and yellow, respectively, while their lightnesses represents whether the relevant bit of $\row(x)-2$ is set (a brighter color corresponds to a set bit). The configuration on the bottom right colors elements according to their target position.}
    \label{fig:radixsortbodies}
\end{figure}

\subsection{Sorting the first row}

We now have a sortable matrix where all columns are body-sorted, which implies that $R$ contains all and only the elements $x$ with $\row(x) = 1$, and we are left with the task of sorting the elements in $R$.
In particular, the arrangement of elements in $R$ can be interpreted as that resulting from the action of some permutation $\pi_R$ on the set $\{1, \dots ,n\}$. Then, our goal becomes that of finding a sequence of row and column rotations whose net result is that of applying $\pi_R^{-1}$ to $R$. This turns out to be more or less challenging depending on the parity of $\pi_R$ and on that of $n$.\footnote{The elements of $R$, in order, coincide with those of the one-line representation of $\pi_R^{-1}$.}

In the rest of this section we build up towards this goal by first showing a way to apply an involution to $R$ while keeping side effects under control, and then using such result to handle general permutations.

\begin{algorithm}[t]
    \caption{Procedure $\SwapPairs$.}\SetKwInput{Input}{Input}
    \SetKwInput{Output}{Output (in place)}
    \Input{An $m \times n$ matrix $A$. A collection $\{ (c_1, c'_1), \dots, (c_k, c'_k)\}$ of pairs of column indices, all distinct. A direction $d \in \{ \texttt{U}, \texttt{D} \}$.}
\Output{The matrix obtained from $A$ by swapping $R[c_h]$ and $R[c'_h]$ and performing a circular shift of the body of $C_h$ in the direction specified by $d$, for all $h=1,\dots,k$.}
    \label{alg:SwapPairs}
  
    \tcp{First phase}
    \For{$r = 0, \dots, n-1$}{
        \For{$j = 1, \dots, k$}{
            \If{$\rho(j, r) = c_j$}{
                \leIf{$d = \texttt{D}$}{rotate $C_{j}$ once}{rotate $C_{j}$ $m-1$ times}
            }
        }
        Rotate $R$\;
    }
    
    \BlankLine
    \tcp{Second phase}
    \For{$r = 0, \dots, n-1$}{
        \For{$j = 1, \dots, k$}{
            \If{$\rho(j, r) = c'_j$}{
                \leIf{$d = \texttt{D}$}{rotate $C_{j}$ $m-1$ times}{rotate $C_{j}$ once}
            }
        }
        Rotate $R$\;
    }

    \BlankLine
    \tcp{Third phase}
    \For{$r = 0, \dots, n-1$}{
        \For{$j = 1, \dots, k$}{
            \If{$\rho(j, r) = c'_j$}{
                \leIf{$d = \texttt{D}$}{rotate $C_{j}$ once}{rotate $C_{j}$ $m-1$ times}
            }
        }
        Rotate $R$\;
    }
\end{algorithm}

\subparagraph*{Swapping pairs of elements in $R$.}
We start by describing an auxiliary procedure that takes as input a collection of $k$ pairs $(c_1, c'_1), \dots, (c_k, c'_k)$ of column indices, where all such indices are distinct, and swaps each element in $R[c_j]$ with that in $R[c'_j]$, for all $j=1,\dots,k$. We name this procedure \SwapPairs and we give its pseudocode in Algorithm~\ref{alg:SwapPairs}.

\SwapPairs performs $O(nk)$ rotations but produces some side effects. More precisely, the procedure needs to use a distinct auxiliary column for each pair $(c_j, c'_j)$ which, for the sake of simplicity, we chose to be column $C_j$. The elements \emph{in the body} of each auxiliary column $C_j$ will undergo a circular shift, which can be either upward (i.e., $C_j[i]$ moves to $C_j[i-1]$ for $i = 3, \dots, m$, and $C_j[2]$ moves to $C_j[m]$) or downward (i.e., $C_j[i]$ moves to $C_j[i+1]$ for $i = 2, \dots, m-1$, and $C_j[m]$ moves to $C_j[2]$). The direction of these rotations is controlled by an additional input parameter $d \in \{ \texttt{U}, \texttt{D} \}$, where $\texttt{U}$ (resp. $\texttt{D}$) denotes an upward (resp.\ downward) rotation.

We note that these side effects are unavoidable as long as the algorithm needs to work for matrices of arbitrary dimensions. Indeed, it is impossible to swap two (or any even number of pairs of) elements without introducing side effects, as this would allow one to change the parity of the permutation induced by positions of the elements in the matrix 
(see Remark~\ref{remark:sortable_characterization}).

\begin{figure}[t]
    \centering
    \includegraphics{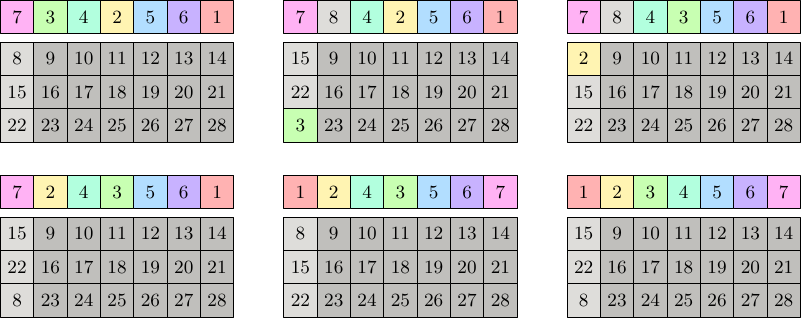}
    \caption{The input configuration to a call to \SwapPairs with $d = \texttt{U}$ swapping $R[2]$ and $R[4]$ (top left), which coincides with the last configuration of Figure~\ref{fig:radixsortbodies}, along with the configurations obtained after the first, second, and third phase (top middle, top right, and bottom left, respectively). Observe how the call causes an upward circular shift of the body of $C_1$ as a side effect. A further call to \SwapPairs with $d = \texttt{D}$ swaps $R[1]$ and $R[7]$ and recovers the original body of $C_1$ (bottom middle). A final call to \SwapPairs with $d = \texttt{U}$ swaps $R[3]$ and $R[4]$, which sorts $R$ but reintroduces the circular shift of the body of $C_1$ compared to the initial configuration (bottom right). Elements $x$ with $\row(x)=1$ are colored according to their target column.}
    \label{fig:first_row}
\end{figure}

The procedure works in three phases.
In the first phase, row $R$ is rotated $n$ times.
Before the generic $r$-th rotation, we check for which indices $j = 1, \dots, k$ the element originally in $R[c_j]$ is now in $R[j]$, i.e., i.e., we check whether $\rho(j, r) = c_j$. For any such index $j$, we rotate $C_j$ one time if $d = \texttt{D}$ or $m-1$ times if $d = \texttt{U}$ (effectively simulating an upward rotation of $C_j$).
Denoting with $S_1$ the sequence of all rotation performed in the first phase, the overall effect of the phase is the following, where $j = 1,\dots, k$ and all elements that not explicitly handled are unaffected:
\begin{itemize}
    \item When $d = \texttt{D}$, $R[c_j] \xrightarrow{S_1} C_j[2]$, $C_j[m] \xrightarrow{S_1} R[c_j]$, and $C_j[i] \xrightarrow{S_1} C_j[i+1]$ for $i = 2, \dots, m-1$;
    \item When $d = \texttt{U}$, $R[c_j] \xrightarrow{S_1} C_j[m]$, $C_j[2] \xrightarrow{S_1} R[c_j]$, and $C_j[i] \xrightarrow{S_1} C_j[i-1]$ for $i = 3,\dots,m$.  
\end{itemize}

In the second phase, $R$ is again rotated $n$ times. 
Before the generic $r$-th rotation, we find all columns $C_j$ that satisfy $\rho(j, r) = c'_j$, and we rotate each such column $m-1$ times if $d = \texttt{D}$, or one time if $d = \texttt{U}$.
The net effect of the sequence $S_2$ of rotations performed during the second phase (on the matrix resulting from the first phase) is the following: 
\begin{itemize}
    \item When $d = \texttt{D}$, $R[c'_j] \xrightarrow{S_2} C_j[m]$, $C_j[2] \xrightarrow{S_2} R[c'_j]$, and $C_j[i] \xrightarrow{S_2} C_j[i-1]$ for $i = 3,\dots,m$.  
    \item When $d = \texttt{U}$, $R[c'_j] \xrightarrow{S_2} C_j[2]$, $C_j[m] \xrightarrow{S_2} R[c'_j]$, and $C_j[i] \xrightarrow{S_2} C_j[i+1]$ for $i = 2, \dots, m-1$;
\end{itemize}
        
The third and final phase is identical to the first one, hence its sequence of rotations $S_3$ produces the same net effect (on the matrix resulting from the second phase) as that of sequence $S_1$. Figure~\ref{fig:first_row} shows some example configurations obtained after the different phases, and after successive executions of \SwapPairs.

\noindent We can now use \SwapPairs to apply any involution to $R$, as the following lemma shows.

\begin{lemma}
    \label{lemma:apply_involution_to_R}
    Consider an $m \times n$ matrix $A$. Given an involution $\pi$ on $\{1, \dots, n\}$ and $d \in \{ \texttt{U}, \texttt{D} \}$, there exists a procedure that uses $O(a_2(\pi) m + n)$ rotations and produces the following effects: (i) permutes $R$ according to $\pi$, and (ii) if $d = \texttt{U}$ (resp.\ $d = \texttt{D}$) performs a cyclic upward (resp.\ downward) rotation of the body of each $C_j$ with $j=1,2,\dots, a_2(\pi)$.
\end{lemma}
\begin{proof}
    Let $\tau_1, \tau_2, \dots, \tau_k$ be the $k = a_2(\pi)$ transpositions of $\pi$, where $\tau_j = (c_j \; c'_j)$.
    We invoke \SwapPairs using the same value of $d$ as the one in the statement and with the $k$ pairs of column indices $(c_j, c'_j)$ for $j=1,\dots,k$. 

    We now prove that such a call has the desired effect. To see that each pair of elements $R[c_j], R[c'_j]$ with $j=1,\dots,k$ is indeed swapped, observe that:
    \begin{itemize}
        \item When $d = \texttt{D}$, $R[c_j] \xrightarrow{S_1} C_j[2] \xrightarrow{S_2} R[c'_j] \xrightarrow{S_3} R[c'_j]$ and $R[c'_j] \xrightarrow{S_1} R[c'_j] \xrightarrow{S_2} C_j[m] \xrightarrow{S_3} R[c_j]$.
        \item When $d = \texttt{U}$, $R[c_j] \xrightarrow{S_1} C_j[m] \xrightarrow{S_2} R[c'_j] \xrightarrow{S_3} R[c'_j]$ and $R[c'_j] \xrightarrow{S_1} R[c'_j] \xrightarrow{S_2} C_j[2] \xrightarrow{S_3} R[c_j]$.
    \end{itemize}

    To see that the body of each column $C_j$ with $j \in \{1, \dots, k\}$ undergoes the claimed circular shift, observe that:
    \begin{itemize}
        \item When $d = \texttt{D}$, $C_j[i] \xrightarrow{S_1} C_j[i+1] \xrightarrow{S_2} C_j[i] \xrightarrow{S_3} C_j[i+1]$ for $i=2,\dots,m-1$, and $C_j[m] \xrightarrow{S_1} R[c_j] \xrightarrow{S_2} R[c_j] \xrightarrow{S_3} C_j[2]$.
        \item When $d = \texttt{U}$, $C_j[i] \xrightarrow{S_1} C_j[i-1] \xrightarrow{S_2} C_j[i] \xrightarrow{S_3} C_j[i-1]$ for $i=3,\dots,m$, and $C_j[2] \xrightarrow{S_1} R[c_j] \xrightarrow{S_2} R[c_j] \xrightarrow{S_3} C_j[m]$.
    \end{itemize}

    To conclude the proof, notice that all elements not explicitly discussed above retain their positions after the rotations in $S_1$, $S_2$, and $S_3$.
\end{proof}

\subparagraph*{Decomposing $\pi_R$ into involutions and transpositions.}

\begin{figure}
    \centering
    \includegraphics[width=\textwidth]{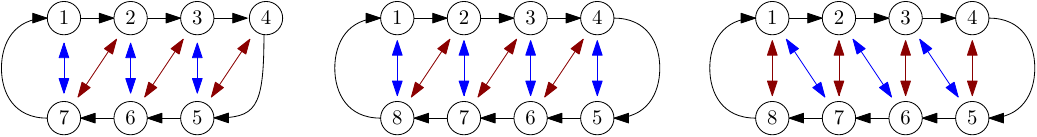}
    \caption{A visualization of the pairs of involutions $\sigma$ and $\upsilon$ (left), $\sigma^+$ and $\upsilon^+$ (middle), and $\sigma^-$ and $\upsilon^-$ (right) for even and odd permutation cycles (in black). Involutions $\sigma$, $\sigma^+$, and $\sigma^-$ are shown in blue and are to be applied first, while involutions $\upsilon$, $\upsilon^+$, and $\upsilon^-$ are shown in red.}
    \label{fig:breaking_cycles}
\end{figure}

We now employ the above result to place all elements in $R$ in their target position, while keeping all columns $C_2, \dots, C_n$ body sorted, and possibly causing circular shift in the body of $C_1$.

It is folklore that any permutation $\pi$ can be written as the product of two involutions (see, e.g., \cite{YangEMR13}). One can also ensure that, when $\pi$ is even, these two involutions have the same number of transpositions, while the odd case requires one more transposition.  We formally state this result in the following two lemmas, in a way that is convenient for the sequel.

\begin{lemma}
    \label{lemma:cycle_to_involutions}
    Let $\pi$ be a permutation that consists of a single non-trivial cycle. 
    If $\pi$ is even, then there exists two involutions $\sigma$, $\upsilon$ such that $\pi = \sigma \upsilon$ and $a_2(\sigma) = a_2(\upsilon)$.
    If $\pi$ is odd, then there exists four involutions $\sigma^+$, $\upsilon^+$, $\sigma^-$, $\upsilon^-$ such that $\pi = \sigma^+ \upsilon^+ = \sigma^- \upsilon^-$, $a_2(\sigma^+) = a_2(\upsilon^+) + 1$, and $a_2(\sigma^-) = a_2(\upsilon^-) - 1$.
\end{lemma}
\begin{proof}
    Assume w.l.o.g.\ that $\pi = (1 \; 2 \; \dots \; \ell)$. The proof is split in two cases, depending on the parity of $\pi$.
    If $\pi$ is even, then $\ell = 2k + 1$ for some positive integer $k$. We choose $\sigma = (1 \; 2k + 1) (2 \; 2k) (3 \; 2k-1) \dots (k \; k+2)$, and $\upsilon = (2 \; 2k+1)(3 \; 2k) (4 \; 2k-1) \dots (k+1 \; k+2)$. Observe that $a_2(\sigma) = a_2(\upsilon) = k$. An example is shown in Figure~\ref{fig:breaking_cycles}~(left).

    If $\pi$ is odd, then $\ell = 2k$ for some positive integer $k$.
    We choose $\sigma^+ = (1 \; 2k)(2 \; 2k - 1)(3 \; 2k - 2) \dots (k \; k+1)$, and $\upsilon^+ = (2 \; 2k)(3 \; 2k - 1) (4 \; 2k-2) \dots (k \; k + 2)$. Observe that $a_2(\sigma^+) = k$ and $a_2(\upsilon^+) = k - 1$. An example is shown in Figure~\ref{fig:breaking_cycles}~(middle).
    We also choose $\sigma^- = (1 \; 2k-1) (2 \; 2k-2) (3 \; 2k-3) \dots (k-1 \; k+1)$ and $\upsilon^- = (1 \; 2k) (2 \; 2k-1) (3 \; 2k-2) \dots (k \; k+1)$. Observe that $a_2(\sigma^-) = k-1$ and $a_2(\upsilon^-) = k$. An example is shown in Figure~\ref{fig:breaking_cycles}~(right).
\end{proof}

\begin{lemma}
    \label{lemma:pemutation_to_involutions}
    Let $\pi$ be a permutation.
    If $\pi$ is even, then there exist two involutions $\sigma, \upsilon$ such that $\pi = \sigma \upsilon $, and $a_2(\sigma) = a_2(\upsilon)$.
    If $\pi$ is odd, then there exist two involutions $\sigma, \upsilon$ and a transposition $\tau$ such that $\pi = \sigma \upsilon \tau $, and $a_2(\sigma) = a_2(\upsilon)$.
\end{lemma}
\begin{proof}
    Let $C^{\text{even}}_1, C^{\text{even}}_2, \dots C^{\text{even}}_h$ and $C^{\text{odd}}_1, C^{\text{odd}}_2, \dots, C^{\text{odd}}_k$ be the non-trivial even and odd cycles of $\pi$, respectively.
    Lemma~\ref{lemma:cycle_to_involutions} shows that we can write each even cycle $C^{\text{even}}_i$ as $C^{\text{even}}_i = \sigma_i \upsilon_i$, where $\sigma_i$ and $\upsilon_i$ are two involutions with the same number of transpositions. We let $\sigma^{\text{even}} = \sigma_1 \sigma_2 \dots \sigma_h$ and $\upsilon^{\text{even}} = \upsilon_1 \upsilon_2 \dots \upsilon_h$.
    Similarly, Lemma~\ref{lemma:cycle_to_involutions} shows that we can write each odd cycle $C^{\text{odd}}_i$ as both $C^{\text{odd}}_i = \sigma^+_i \upsilon^+_i$ and $C^{\text{odd}}_i = \sigma^-_i \upsilon^-_i$, where each of $\sigma^+_i$, $\upsilon^+_i$, $\sigma^-_i$, and $\upsilon^-_i$ is an involution, $a_2(\sigma^+_i) = a_2(\upsilon^+_i) + 1$, and $a_2(\sigma^-_i) = a_2(\upsilon^-_i) - 1$.

    The rest of the proof depends on the parity of $\pi$.
    We consider the case in which $\pi$ is even first, which implies that $k$ is also even. 
    We let $\sigma^{\text{odd}} = \sigma^-_1 \sigma^+_2 \sigma^-_3 \sigma^+_4 \sigma^-_5 \dots \sigma^+_k$ and $\upsilon^{\text{odd}} = \upsilon^-_1 \upsilon^+_2 \upsilon^-_3 \upsilon^+_4 \upsilon^-_5 \dots \upsilon^+_k$ (notice the alternating superscripts).
    The claim follows by choosing $\sigma = \sigma^{\text{even}} \sigma^{\text{odd}}$ and $\upsilon = \upsilon^{\text{even}} \upsilon^{\text{odd}}$.

    We now consider the case in which $\pi$ is odd, which implies that $k$ is also odd.
    We let $\tau$ be an arbitrary transposition of $\upsilon^-_k$, and we define $\upsilon'_k = \upsilon^-_k \tau$ as the involution obtained by removing $\tau$ from $\upsilon^-_k$ (observe that $\tau$ is its own inverse), so that $a_2(\sigma^-_k) = a_2(\upsilon^-_k) - 1 = a_2(\upsilon'_k)$.  We let $\sigma^{\text{odd}} = \sigma^-_1 \sigma^+_2 \sigma^-_3 \sigma^+_4 \sigma^-_5 \dots \sigma^+_{k-1} \sigma^-_k$ and $\upsilon^{\text{odd}} = \upsilon^-_1 \upsilon^+_2 \upsilon^-_3 \upsilon^+_4 \upsilon^-_5 \dots \upsilon^+_{k-1} \upsilon'_k$ (notice the alternating superscripts).
   The claim follows by choosing $\sigma = \sigma^{\text{even}} \sigma^{\text{odd}}$, $\upsilon = \upsilon^{\text{even}} \upsilon^{\text{odd}}$, and $\tau$ as defined above.
\end{proof}

\subparagraph*{Applying $\pi^{-1}_R$ to $R$.}

We can finally show how the elements in $R$ can be sorted, i.e., permuted according to $\pi^{-1}_R$.
We start with a simplification: we assume w.l.o.g.\ that if $\pi_R$ is odd, then $m$ is even.
Indeed, the case in which $\pi_R$ is odd and $m$ is odd can be avoided by noticing that, whenever this happens, $n$ must be even since otherwise the matrix would not be sortable (see Remark~\ref{remark:sortable_characterization}). Hence, we can reduce to the case in which $\pi_R$ is even by rotating $R$.

We proceed differently depending on the parity of $\pi_R$.
If $\pi_R$ is even, we use Lemma~\ref{lemma:pemutation_to_involutions} to write $\pi_R^{-1}$ as the product $\sigma \upsilon$ of two involutions $\sigma, \upsilon$ with $a_2(\sigma) = a_2(\upsilon)$.
Then we execute procedure \SwapPairs twice: the first call applies $\sigma$ to $R$ and uses $d=\texttt{U}$, while the second applies $\upsilon$ to the resulting configuration and uses $d=\texttt{D}$. Invoking Lemma~\ref{lemma:apply_involution_to_R} twice shows that all rotations caused to the column bodies cancel out, while the elements originally in $R$ get permuted according to $\pi_R^{-1} = \sigma \upsilon$, i.e., they end up in the sorted configuration.

\begin{figure}[t]
    \centering
    \includegraphics{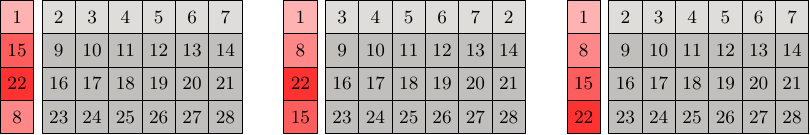}
    \caption{The configuration resulting from applying $\sigma = (2 \; 4)$, $\upsilon = (1 \; 7)$, and $\tau = (3 \; 4)$ to $R$ starting from the configuration in the top left of Figure~\ref{fig:first_row} (top left). Observe how the swaps implementing $\sigma$, $\upsilon$, and $\tau$ are exactly those shown in Figure~\ref{fig:first_row}, and how the body of $C_1$ undergoes a circular shift.
    The correct order of the elements in $C_1$ is restored by two consecutive calls to \SwapPairs (on the transposed matrix) applying involutions $\sigma_1=(2 \; 4)$ (top middle) and $\upsilon_1 = (3 \; 4)$ (top right). These latter two calls first introduce and then undo a circular shift of $R[2], R[3], \dots, R[n]$. The elements in $C_1$ are colored in red, where different lightnesses correspond to different target rows. A small gap has been added to the right of column $C_1$ for improved visual clarity.}
    \label{fig:first_column}
\end{figure}

The only remaining case is the one in which $\pi_R$ is odd and $m$ is even.
In this case, we use Lemma~\ref{lemma:pemutation_to_involutions} to write $\pi_R^{-1}$ as the product $\sigma \upsilon \tau$ of two involutions $\sigma, \upsilon$ with $a_2(\sigma) = a_2(\upsilon)$, and a transposition $\tau$.
We call procedure \SwapPairs five times. The first two calls apply involutions $\sigma$ and $\upsilon$ to $R$, in this order, with $d=\texttt{U}$ and $d=\texttt{D}$ respectively. Using Lemma~\ref{lemma:apply_involution_to_R} twice, we have that the combined effect of these two calls to $\SwapPairs$ leaves the positions of the elements not in $R$ unaffected.
The third call swaps the elements involved in the transposition $\tau$, while causing a circular shift in body of $C_1$ (the direction of such shift is not important).
Since $m$ is even, the number of elements in the body of $C_1$ is odd, and the permutation $\pi_1$ describing their circular shift is even, hence it can be written as the product $\sigma_1 \tau_1$ of two involutions $\sigma_1 \tau_1$ with the same number of transpositions.
Since we can equivalently consider the transposed version $A^T$ of matrix $A$, we use the fourth and fifth calls to \SwapPairs to apply $\sigma_1$ and $\tau_1$, in this order, to the first row of $A^T$, i.e., to column $C_1$ of $A$. By using opposite values of $d$ in these two final calls, we ensure any side effects introduced by the fourth call are undone by the fifth and last call. An example is provided in Figure~\ref{fig:first_column}.
The above discussion is summarized in the following lemma.

\begin{lemma}
    \label{lemma:sorting_R}
    Given an $m \times n$ matrix $A$ in which each column is body-sorted, there exists a procedure that performs $O(mn)$ rotations and sorts $A$.
\end{lemma}

Algorithm~\ref{alg:sorting} shows the final pseudocode for solving the Torus Puzzle when $m \le n$ (which can be assumed w.l.o.g.), where lines from \ref{ln:sorting_R_start} onward handle sorting $R$ as described above.

\noindent The main result of this paper follows from combining Corollary~\ref{cor:near_full} and Lemmas~\ref{lemma:body_full}, \ref{lemma:body-sorted}, and \ref{lemma:sorting_R}. 

\begin{theorem}
    There exists an algorithm that, given a sortable $m \times n$ instance of the Torus Puzzle, sorts it by performing $O(mn \cdot \log \max\{m,n\})$ rotations, where each rotation is either a unit downward rotation of a column or unit rightward rotation of a row.
\end{theorem}

\begin{algorithm}[t]
    \caption{Algorithm for solving the Torus Puzzle.}
    \label{alg:sorting}
    \SetKwInput{Input}{Input}
    \SetKwInput{Output}{Output (in place)}
    \Input{A sortable $m \times n$ matrix $A$ with $m \le n$ (consider $A^T$ for $m > n$).}
    \Output{The sorted version of $A$.}
    \BlankLine

    Call \FillColumns on $A$\;
    Call \FloatMinimums on $A$\;
    \BlankLine

    $k \gets \big\lfloor \frac{n}{m-1} \big\rfloor$\;
    \For{$h = 0, 1, \dots, \lceil \frac{n}{k} \rceil - 1$}{
        Call \RadixSortBodies on columns $hk + 1, hk+2, \dots, \min\{(h+1)k, n\}$\;
    }
    \BlankLine

    \If{$n$ is even $\wedge$ the permutation induced by $R$ is odd\label{ln:sorting_R_start}}{
        Rotate $R$\;
    }

    \BlankLine

    $\pi_R \gets $ permutation induced by $R$\;
    \If{$\pi_R$ is even}{
        $\sigma, \upsilon \gets$ Two involutions s.t.\ $\pi^{-1}_R = \sigma \upsilon$ and $a_2(\sigma) = a_2(\upsilon)$\;
        Call $\SwapPairs$ with $d=\texttt{U}$ to apply $\sigma$ to $R$\;
        Call $\SwapPairs$ with $d=\texttt{D}$ to apply $\upsilon$ to $R$\;
    }
    \Else{
        \label{ln:else_pi_r_odd}
        $\sigma, \upsilon, \tau \gets$ Two involutions and a transposition s.t.\ $\pi^{-1}_R = \sigma \upsilon \tau$ and $a_2(\sigma) = a_2(\upsilon)$\;
        Call $\SwapPairs$ with $d=\texttt{U}$ to apply $\sigma$ to $R$\;
        Call $\SwapPairs$ with $d=\texttt{D}$ to apply $\upsilon$ to $R$\;
        Call $\SwapPairs$ with $d=\texttt{U}$ to apply $\tau$ to $R$\;
        \BlankLine
        $\pi_1 \gets $ permutation induced by $C_1$\;
        $\sigma_1, \upsilon_1 \gets$ Two involutions s.t.\ $\pi_1^{-1} = \sigma_1 \upsilon_1$ and $a_2(\sigma_1) = a_2(\upsilon_1)$\;
        Call $\SwapPairs$ with $d=\texttt{U}$ to apply $\sigma_1$ to $C_1$\;
        Call $\SwapPairs$ with $d=\texttt{D}$ to apply $\upsilon_1$ to $C_1$\;
    }
\end{algorithm}

\section{Concluding remarks and open problems}
\label{sec:conclusions}

\subparagraph*{Remarks on the time-complexity of our algorithm.} 
An implementation of the procedures used in our algorithm which closely follows the provided pseudocodes can result in a running time that is larger than our $O(mn \cdot  \log \max\{m, n\})$ upper bound on the push number.

Here, we briefly sketch how the algorithm can be implemented to also run in time $O(mn \cdot \log \max\{m, n\})$. We start by observing that, when $m \le n$, the only row rotated by procedures \FillColumns, \FloatMinimums, \RadixSortBodies, and \SwapPairs is $R$. 
Then, a data structure that maintains the current configuration of $A$ while allowing constant-time element access, and constant-time (not necessarily unit) rotations of $R$ and $C_1, \dots, C_j$ can be obtained by storing $R$ and each of the column bodies in circular arrays equipped with pointers to their first elements.
Moreover, for a given column $C_j$ and at any point during the execution of \FillColumns, \FloatMinimums, each (sub-)phase of \RadixSortBodies, and each of the three phases of \SwapPairs, it is easy to know when the next rotation of $C_j$ will occur. The idea is that of introducing a concept of time, which is determined by the number of rotations of $R$ performed so far, while maintaining a min-heap $H$ storing each of column index $j$ with a key equal to the time of the next rotation of $C_j$ (if any). Since the next rotation of a column (if it exists) is never more than $n$ time-instants away, $H$ can be implemented to support insertions in constant-time, and extractions of the minimum in time proportional to the number of time-instants between the current time and the one given by the key of the extracted element.
With these ingredients, all the above procedures can be implemented to run in the same asymptotic time as their respective upper bounds on the number of rotations, and the same follows for the whole sorting algorithm (which, for $m \le n$, repeatedly invokes these procedure while transposing $A$ at most once).

\subparagraph*{Reducing the number of distinct rotation types.}
When our algorithm is executed on an $m \times n$ input matrix with $m \le n$, all rotations involve $R$ or the columns $C_1, \dots, C_n$ up until the very end, when \SwapPairs might be invoked to apply the permutation $\pi_1$ to $C_1$ (i.e., to the first row of the transposed version of $A$). This last step may rotate rows other than $R$.

One might wonder whether it is possible to avoid rotating these rows, since doing so would result in an algorithm that sorts any sortable matrix using no more than $m+1$ distinct rotation types, which is also a lower bound.\footnote{To see that $m$ rotation types are not enough in general, notice that at least one row and at least one column must necessarily be available for rotations. Then, for any such choice of $m$ rotations types, there is always a row-column pair such that neither the row nor the column can be rotated, which implies that the element on their intersection cannot be moved from its initial position.}
This can be by using the following alternative handing of the case in which $\pi_R$ is odd (i.e., the ``else'' branch starting at line~\ref{ln:else_pi_r_odd} of Algorithm~\ref{alg:sorting}): after decomposing $\pi_R^{-1}$ into the product $\sigma\upsilon\tau$, call \SwapPairs with alternating choices of $d$ to apply $\sigma$, $\upsilon$, $\tau$, and $(1 \; n)$ to $R$, in this order.
This results in a matrix in which every element is in its target position, except for $R[1]$ and $R[n]$, which need to be swapped.
In Appendix~\ref{apx:swap}, 
we show a sequence of $O(mn)$ unit rotation that performs  exactly this task when $m$ is even (as in the case of interest), while only rotating $R$ and $C_1$.

\subparagraph*{Handling arbitrary rotation directions.}
While the description of our algorithm assumes that all columns can only be rotated downwards, it is easy to check that all our procedures work with either no or with minor modifications even when each column can only be rotated in a direction specified as part of the input (where different columns can have different directions). 
This implies that our upper bound holds even when both rows and columns have input-specified rotation directions. The easiest way to see this is that of combining the previous discussion showing that the only row that needs to be rotated is $R$, with symmetry arguments which allow to equivalently consider a transposed or mirrored version of $A$.\footnote{In this way, one can always assume to be dealing with a matrix with at least as many columns as rows, and where unit rightward rotations of the first row are allowed.}

\subparagraph*{Open problems.}
Following our results, the known lower and upper bounds on the push number differ by a multiplicative $\Theta(\log \max\{m, n\})$ factor. It would be interesting to close this gap by providing asymptotically tight bounds.

Moreover, to the best of our knowledge, it is still unknown whether the optimization version of the Torus Puzzle (that asks to return the shortest sequence of rotations sorting the input matrix) is solvable in polynomial time,  both for the case unit rotations and for the more permissive case of compound rotations. Likewise, no polynomial-time approximation algorithm with an approximation factor better than the trivial $\widetilde{O}(mn)$ is currently known.

\bibliographystyle{plainurl}
\bibliography{bibliography.bib}


\clearpage
\appendix

\section{Swapping two adjacent elements using \texorpdfstring{$O(mn)$}{O(mn)} unit rotations}
\label{apx:swap}

In this section we assume that at least one of $n$ and $m$ is even, and we show how a pair of adjacent elements can be swapped using $O(mn)$ rotations, without affecting the position of any other element.\footnote{Recall that when $n$ and $m$ are both odd, no sequence of rotation can change the parity of the permutation induced by the current configuration, hence swapping a single pair of elements is impossible (see Remark~\ref{remark:sortable_characterization}).}

In the following we only focus on swapping $R[1]$ and $R[n]$.  However, it is easy to see that this actually shows how to swap any pair of adjacent elements using the same asymptotic number of rotations once one observes that: (i) by symmetry, any sequence of rotations swapping $R[1]$ and $R[n]$ also applies to any pair of adjacent positions of the first row (after a suitable renaming of the columns), and (ii) to swap two adjacent elements from the same column $C_j$, one can ``unload'' the elements into $R$, swap them, and then ``load'' the swapped elements back into $C_j$ (which also restores the elements of $R$, in their original order).

This improves over the sequence provided by \cite{AmanoKKKNOTY12}, which perform a similar task using $\Omega(\max\{ m, n \}^2)$ moves in the worst case, and provides explicit sequences for Greenwell's observation that two adjacent elements of an $m \times n$ matrix can be swapped when at least one of $m$ and $n$ is even \cite{Greenwell1999}.

For a sequence of rotations $S$ and a non-negative integer $k$, we use $S^k$ to denote the sequence corresponding to $k$ consecutive repetitions of $S$. We start by handling the case in which $n$ is even.

\begin{lemma}
    \label{lemma:swap_n_even}
    Consider an $m \times n$ matrix $A$ with $m,n \ge 2$ and even $n$.
    The sequence of rotations
    $C_1 R^2 C_1^{m-1} R (C_1 R C_1^{m-1} R)^{\frac{n}{2}-1}$ swaps  $R[1]$ and $R[n]$ while leaving the position of all other elements unaffected.
\end{lemma}
\begin{proof}
    Let $X = C_1 R^2 C_1^{m-1} R$ and $Y = C_1 R C_1^{m-1} R$, so that the sequence can be written as $X Y^{\frac{n}{2} - 1}$.
    We start by handling the case $n=2$. In this case, the sequence $X = C_1 R^2 C_1^{m-1} R$ is equivalent to $C_1 C_1^{m-1} R$ which, in turn, is equivalent to $R$. Moreover, $Y^{\frac{n}{2} - 1}$ is the empty sequence. Then, the sequence of rotations $X Y^{\frac{n}{2} -1}$ yields the same configuration as rotation $R$, which is easily checked to have the claimed effect. Hence, in the rest of the proof we will assume $n \ge 4$.

    Since neither $X$ nor $Y$ ever rotate any column $C_j$ with $j \ge 2$, we conclude that the position of all elements in the bodies of $C_2, \dots, C_n$ are not affected by $XY^{\frac{n}{2}-1}$.
    Moreover, since the effect of $C_1^{m-1}$ is that of performing a cyclic upward rotation of column $C_1$, we have that both $X$ and $Y$ perform one downward rotation of $C_1$ followed by (the equivalent of) one upward rotation, hence they do not affect the position of any element $C_1[i]$ for $i \in \{2, \dots, n-1\}$ and the same holds for $XY^{\frac{n}{2}-1}$. 

    We now argue that $XY^{\frac{n}{2}-1}$ places the elements initially in $R$ in the right final positions, which suffices to imply the claim (indeed, the only remaining element $C_1[m]$ must be correctly placed by a process of elimination). We start with some observations:
    \begin{itemize}
        \item $R[1] \xrightarrow{X} R[2]$,
            $R[i] \xrightarrow{X} R[i+3]$ for $i \in \{2, \dots, n-3\}$,
            $R[n-2] \xrightarrow{X} R[1]$,  $R[n-1] \xrightarrow{X} C_1[m]$, and $R[n] \xrightarrow{X} R[3]$.
        \item We have $R[i] \xrightarrow{Y} R[i+2]$ for any $i \in \{2,\dots, n-2\}$. Hence, for any $i \in \{2, \dots, n\}$ and any integer $h \le \frac{n-i}{2}$, we have $R[i] \xrightarrow{Y^h} R[i + 2h]$.
    \end{itemize}

    \noindent The rest of the proof depends on the values of $i$ and $n$.
    \begin{itemize}
        \item If $i = 1$, then $R[1] \xrightarrow{X} R[2] \xrightarrow{Y^{\frac{n}{2}-1}} R[n]$. 

        \item If $2 \le i \le n-4$ and $i$ is even, let $h = \frac{n-i-4}{2}$, $k = \frac{i-2}{2}$, and observe that $XY^{\frac{n}{2} - 1} = X Y^{h} YY Y^{k}$.
            We have $R[i]\xrightarrow{X} R_i[i+3] \xrightarrow{X} R[n - 1]$ since $i + 3 + 2h = n - 1$, $R[n-1] \xrightarrow{Y} R[1] \xrightarrow{Y} R[2]$, and $R[2] \xrightarrow{Y^k} R[i]$ since $2 + 2k = i$.  

        \item If $3 \le i \le n-3$ and $i$ is odd, let $h = \frac{n-i-3}{2}$, $k = \frac{i-3}{2}$, and observe that $XY^{\frac{n}{2} - 1} = X Y^{h} YY Y^{k}$.
            We have $R[i] \xrightarrow{X} R[i+3] \xrightarrow{Y^h} R[n]$ since $i + 3 + 2h = n$, 
            $R[n] \xrightarrow{Y} C_1[m] \xrightarrow{Y} R[3]$, and
            $R[3] \xrightarrow{Y^k} R[i]$ since $3+2k = i$. 

        \item If $i = n-2$, let $h = \frac{n-4}{2}$ and observe that $XY^{\frac{n}{2} - 1} = XYY^h$. We have $R[n-2] \xrightarrow{X} R[1] \xrightarrow{Y} R[2] R[2] \xrightarrow{Y^h} R[n-2]$ since $2 + 2h = n-2$. It follows that $R[n-2] \xrightarrow{XYY^h} R[n-2]$. 

        \item If $3 \le i = n-1$, let $h = \frac{n-4}{2}$ and observe that $XY^{\frac{n}{2} - 1} = XYY^h$.
            We have $R[n-1] \xrightarrow{X} C_1[m] \xrightarrow{Y} R[3] \xrightarrow{Y^h} R[n-1]$ since $3 + 2h = n-1$. 

        \item If $i=n$, let $h=\frac{n-4}{2}$ and observe that $XY^{\frac{n}{2}-1} = XY^hY$. We have $R[n] \xrightarrow{X} R[3] \xrightarrow{Y^h} R[n-1]$ since $3 + 2h = n-1$, and $R[n-1] \xrightarrow{Y} R[1]$. \qedhere
    \end{itemize}
\end{proof}

\noindent We now consider the case in which $m$ is even.

\begin{lemma}
Consider an $m \times n$ matrix $A$ with $m,n \ge 2$ and even $m$.
    The sequence of rotations
    $C_1^{m-1} R^2 C_1^2 R^{n-1} C_1 (R C_1 R^{n-1} C_1)^{\frac{m}{2}-1} R^{n-1} C_1$ swaps $R[1]$ and $R[n]$ while leaving the position of all other elements unaffected.
\end{lemma}
\begin{proof}
    Define the sequence
    $X = R C_1^2 R^{n-1} C_1 (R C_1 R^{n-1} C_1)^{\frac{m}{2}-1}$
    and the sequence
    $X^T = C_1 R^2 C_1^{n-1} R (C_1 R C_1^{n-1} R)^{\frac{m}{2}-1}$,
    which has been obtained from $X$ by replacing each occurrence of $R$ with $C_1$ and vice versa.
    The sequence in the claim corresponds to $C_1^{m-1} R X R^{n-1} C_1$.

    Consider the $n \times m$ matrix $A^T$ obtained by transposing $A$, and let $R^T$ be the first row of $A^T$.
    By Lemma~\ref{lemma:swap_n_even}, applying the sequence of rotations $X^T$ to $A^T$ swaps $R^T[1]$ and $R^T[m]$ without modifying the position of the other elements (notice that the roles of $n$ and $m$ are exchanged in $X^T$ compared to the statement of Lemma~\ref{lemma:swap_n_even} to account for the dimensions of $A^T$). Then, applying $X$ to $A$ has the sole effect of swapping $C_1[1]$ and $C_1[m]$. 

    We now argue that $C_1^{m-1} R X R^{n-1} C_1$ swaps $R[1]$ and $R[n]$. Indeed, using the fact that $R[1]$ and $C_1[1]$ refer to the same element, we have:
    \begin{itemize}
        \item $R[1] \xrightarrow{C_1^{m-1}} C_1[m] \xrightarrow{R} C_1[m] \xrightarrow{X} C_1[1] \xrightarrow{R^{n-1}} R[n] \xrightarrow{C_1} R[n]$;

        \item $R[n] \xrightarrow{C_1^{m-1}} R[n] \xrightarrow{R} R[1] \xrightarrow{X} C_1[m] \xrightarrow{R^{n-1}} C_1[m] \xrightarrow{C_1} C_1[1]$.
    \end{itemize}

    To conclude the proof, we now show that the final positions of all other elements are not affected by the sequence of rotations. This is trivial for elements that are neither in $R$ nor in $C_1$ (since neither their rows nor their columns are ever rotated), while for the remaining elements we distinguish the following cases:
    \begin{itemize}
        \item $C_1[2] \xrightarrow{C_1^{m-1}} C_1[1] \xrightarrow{R} R[2] \xrightarrow{X} R[2] \xrightarrow{R^{n-1}} R[1] \xrightarrow{C_1} C_1[2]$;
    
        \item If $i \in \{3, \dots, m\}$, $C_1[i] \xrightarrow{C_1^{m-1}} C_1[i-1] \xrightarrow{R} C_1[i-1] \xrightarrow{X} C_1[i-1] \xrightarrow{R^{n-1}} C_1[i-1] \xrightarrow{C_1}[i]$;

        \item If $j \in \{2, \dots, n-1\}$, $R[j] \xrightarrow{C_1^{m-1}} R[j] \xrightarrow{R} R[j+1] \xrightarrow{X} R[j+1] \xrightarrow{R^{n-1}} R[j] \xrightarrow{C_1} R[j]$. \qedhere
    \end{itemize}
\end{proof}

\end{document}